\documentclass[12pt]{amsart}
\usepackage{amssymb,latexsym}
\theoremstyle{plain}
\numberwithin{equation}{section}
\newtheorem{theorem}{Theorem}
\newtheorem{lemma}{Lemma}
\newtheorem{proposition}{Proposition}

\newtheorem{definition}{Definition}
\begin{document}

\title{\textbf{Quadratic  control of quantum processes}}

\author{Luigi Accardi}
\address{Centro Vito Volterra, Universit\'{a} di Roma TorVergata\\
            Via di TorVergata, 00133 Roma, Italy}
\email{volterra@volterra.mat.uniroma2.it}
\author{Andreas Boukas}
\address{Department of Mathematics, American College of Greece\\
 Aghia Paraskevi, Athens 15342, Greece}
\email{gxk-personnel@ath.forthnet.gr}
\thanks{ The second author wishes to thank Professor Luigi Accardi  for his support and for the hospitality of the Centro Vito Volterra of the Universita di Roma TorVergata on several occasions.}


\date{}

\begin{abstract}
Within the framework of the Accardi-Fagnola-Quaegebeur (AFQ) representation free calculus of \cite{b}, we consider the problem of controlling the size of a quantum stochastic flow generated by a unitary stochastic evolution affected by quantum noise. In the case when the evolution is driven by first order white noise (which includes quantum Brownian motion) the control is shown to be given in terms of the solution of an algebraic Riccati equation.
 This is done by first solving the problem of controlling (by minimizing an associated quadratic performance criterion) a stochastic process whose evolution is described by a stochastic differential equation of the type considerd in \cite{b}. The solution is given as a feedback control law in terms of the solution of a stochastic Riccati equation.  \end{abstract}

\maketitle

\section{\textbf{Introduction}}

In classical quantum mechanics the time-evolution of an observable $X$ (i.e a self-adjoint operator on the wave function Hilbert space) is described at each time $t \geq 0$  by a new observable $j_t(X)=U^*(t)\,X\,U(t)$ where $U(t)=e^{-itH}$ and $H$ is a self-adjoint operator on the wave function space. In this case the unitary process $U=\{U(t)\,/\,t\geq0\}$ satisfies the differential equation

\begin{eqnarray*}
dU(t)=-iH\,U(t)\,dt,\,\,U(0)=I
\end{eqnarray*}

and the dynamics of the quantum flow $j_t(X)$ is described by the Heisenberg equation

\begin{eqnarray*}
dj_t(X)=i[H,j_t(X)]\,dt,\,\,j_0(X)=X
\end{eqnarray*}

where

\begin{eqnarray*}
[H,j_t(X)]=H\,j_t(X)-j_t(X)\,H
\end{eqnarray*}

   In the general case when the  system is affected by quantum noise, the equation satisfied by the unitary process $U$ is  a quantum stochastic differential equation  driven by that noise (see e.g \cite{o, l1}) and the corresponding equation for the quantum flow  $j_t(X)$  is interpreted as the  Heisenberg picture of the Schr\"{o}dinger equation in the presence of noise or as a quantum probabilistic analogue of the  Langevin equation. The problem of determining conditions on the coefficients of such equations  that guarantee the unitarity of $U$ has been studied extensively for various quantum noises.

In this paper we consider the problem of controlling the size of  a  quantum flow by minimizing  the  $L^2$-type performance functional $J_{\xi,T}(\cdot)$ defined in  the following

\begin{definition}
Let $T$ be a finite time, $\xi$  a vector in a dense subset of the space on which the flow lives, $M$  a positive operator on the same space, and \break $u=\{u_t\,/\,t \geq 0 \}$  a control process appearing additively in the deterministic part of the quantum stochastic differential equation satisfied by the unitary process  \break $U=\{U(t)\,/\,t\geq0\}$ (see Section 5 below). Define

\begin{eqnarray}
J_{\xi,T}(u)=\int_0^T\,[\|j_t(X)\xi\|^2+\|u_t \xi\|^2\, ]\,dt+<\xi,j_T(M)\xi>
\end{eqnarray}

where $\langle \cdot , \cdot \rangle $ is the inner product in $H$ and $\| \cdot \|$ is the corresponding norm.

\end{definition}

If we wish to allow for flows that do not necessarily correspond to classical quantum mechanical observables,  by allowing for non self-adjoint $X$, we replace $\|j_t(X)\xi\|^2$ in (1.1) by 
$<\xi,j_t(X^*X)\xi>$.

  We interpret the first term on the right hand side of (1.1) as a measure of the size of the flow over $[0,T]$, the second as a measure of the control effort over $[0,T]$ and the third as a "penalty" for allowing the evolution to go on for a long time.

  The  functional  $J_{\xi,T}(\cdot)$ of Definition 1 , which we propose as suitable for the evaluation of the performance of a quantum flow,  will be shown in Section 5 to be derived from a quantum analogue of the classical quadratic performance criterion for operator processes $X=\{X(t)\,/\,t \geq 0\}$,  described in Section 3 as solutions of Hilbert space stochastic differential equations, in the case when $X(t)=U(t)$ is a unitary operator for each $t \geq 0$. The operators $X(t)$ are not necessarily self-adjoint so they do not in general correspond to quantum mechanical observables in the classical sense unless special assumptions are made on the coefficients of the defining quantum stochastic differential equations. However, this is not a problem since the conrol problem for the $X(t)\,'s$ is to be used as a passage to the solution of the control problem for the corresponding quantum flow  $j_t(\cdot )$ which does consist of quantum mechanical observables. For more on quantum stochastic flows see \cite{d,d1}.

In quantum probability, starting with an operator representation of a Lie algebra, operator analogues and generalizations of classical stochastic noise processes (such as Brownian motion, the exponential process, the Poisson process e.t.c.) as well as purely quantum noises such as the square of white noise ( see \cite{c}) can be constructed  ( see e.g \cite{i,n,o}). 

The quantum stochastic calculi constructed in order to study evolutions driven by these operator noises were dependent on the particular representation and led to analytic difficulties such as the unboundedness of  solutions of stochastic differential equations, the non-invariance of their domain e.t.c. These problems were removed by the introduction of the AFQ representation free calculus of \cite{b} which provided the analytic and topological framework for a unified treatment of  quantum noises and which is now in standard use. In the following section  we provide  a brief review of the AFQ calculus for quick reference.

\medskip
 \section{\textbf{Review of the representation free calculus}}
\medskip

Let $H$ be a complex separable Hilbert space, $B(H)$ the algebra of all bounded linear operators on $H$, $D$ a total subset of $H$, $(A_{t]})_{t \geq 0}$ an increasing family of $W^*$-algebras of operators on $H$, $A$ a $W^*$-algebra of operators on $H$ such that $A_{t]} \subseteq A$ for all $t \geq 0$, $A_{t]}^{\prime}$ the commutant of  $A_{t]}$ in $B(H)$, $H_{t]}(\xi)$ the closure for each $\xi \in D$ of the subspace $[A_{t]}]=\{\alpha \xi \,/\, \alpha \in A_{t]}\}$, $e_{t]}^{\xi}$ the orthogonal projection onto $H_{t]}(\xi)$, $L(D;H)$ the vector space of all linear operators $F$ with domain containing $D$ and such that the domain of the adjoint operator $F^*$ also includes $D$.

A \textit{random variable} is an element of  $L(D;H)$. A \textit{stochastic process} in $H$ is a family $F=\{F_t\,/\,t \geq 0\}$ of random variables such that for each $\eta \in D$ the map $t \rightarrow F_t \eta$ is Borel measurable. If  $F_t \in B(H)$ for each $t \geq 0$, and  $\sup_{0 \leq t \leq T} \|F_t\| < + \infty$ for each $T < + \infty$, then the process $F$ is called \textit{locally bounded}. If  $F_t \geq 0$ for each $t \geq 0$ then $F$ is \textit{positive}.

Let $D_t^{\prime}$ denote the linear span of $A_{t]}^{\prime}D$. An operator $F$ is \textit{t-adapted} to $A_{t]}$ if $dom(F)=D_t^{\prime} \subseteq dom(F^*)$ and, for all $\alpha_t^{\prime} \in  A_{t]}^{\prime}$ and $\xi \in D$,  $F\alpha_t^{\prime}\xi=\alpha_t^{\prime}F\xi$ and $F^*\alpha_t^{\prime}\xi=\alpha_t^{\prime}F^*\xi$. Strong limits of sequences of $t$-adapted operators are  $t$-adapted. A stochastic process $F$ is \textit{adapted to the filtration} $\{A_{t]}\,/\,t \geq 0\}$ if $F_t$ is adapted for all $t \geq 0$, and it is \textit{simple} if $F_t=\sum_{k=1}^n \chi_{[t_k,t_{k+1})}(t) F_{t_k}$ for some finite integer $n$ and $0 \leq t_0 < t_1 <...<t_{n+1} <+ \infty$.

An \textit{additive process} is a family $M=\{M(s,t)\,/\,0 \leq s \leq t\}$ of random variables such that for all $s \leq t$ the operator $M(s,t)$ is $t$-adapted and, for all $r, s, t$ with $r \leq s \leq t$, $M(r,t)=M(r,s)+M(s,t)$ and $M(t,t)=0$ on $D$. To every additive process $M$ we associate the adapted process $M(t)=M(0,t)$ and conversely to every adapted process $\{M(t)\,/\, t \geq 0\}$ we associate the additive process $M(s,t)=M(t)-M(s)$. An additive process $M$ is \textit{regular} if, for all $\xi \in D$ and  $r \leq s \leq t$, $H_{r]}(\xi) \subseteq dom(\overline {M^\#} (s,t))$ and $M^*(s,t)D \subseteq  D_s^{\prime}$, where $\overline{ M}$ denotes the closure of $M$ and $M^\#$ denotes either $M$ or $M^*$.  

If $M$ is a regular additive process and $F$ is a simple adapted process then the \textit{left (resp. right) stochastic integral } of $F$ with respect to $M$ over the interval $[0,t]$ is defined as an operator on $D_t^{\prime}$ by

\[
\int_0^t\,dM_s\,F_s=\sum_{k=1}^n\overline{ M} (t_k \wedge t,t_{k+1} \wedge t_0)\,F_{t_k}|_{D_t^{\prime}}
\]

\[
\mbox{ (resp. }\int_0^t\,F_s \, dM_s=\sum_{k=1}^n F_{t_k}\,
M (t_k \wedge t,t_{k+1} \wedge t_0))
\]

An additive regular process is an \textit{integrator of scalar type} if for each $\xi \in D$ there exists a finite set $J(\xi) \subseteq D$ such that for each simple process $F$ and $t \geq 0$

\[
\|\int_0^t\,dM_s\,F_s \xi\|^2 \leq c_{t,\xi} \int_0^t \, d \mu_{\xi}(s)\,\sum_{\eta \in J(\xi)}\|F_s\eta\|^2
\]

and

\[
\|\int_0^t\,F_s^* \,dM_s^*\xi\|^2 \leq c_{t,\xi} \int_0^t \, d \mu_{\xi}(s)\,\sum_{\eta \in J(\xi)}\|F_s^*\eta\|^2
\]

and also, for all $\eta \in J(\xi)$, 

\[
J(\eta) \subseteq J(\xi)
\]

where $c_{t,\xi} \geq 0$  and $\mu_{\xi}$ is a positive, locally finite, non atomic measure. 

If $M$ is an integrator of scalar type and, for all $\xi \in D$, $\mu_{\xi}$  is absolutely continuous with respect to Lebesgue measure then the stochastic integral with respect to $M$ can  be extended by continuity to processes $F \in L^2_{loc}([0,+\infty),dM)$, the space of all adapted processes $F$ with the topology induced by the seminorms

\[
\|F\|^2_{\eta,t,\mu_{\xi}}=\int_0^t\,\|F_s\eta\|^2\,d\mu_{\xi}(s)
\]

such that for all $\xi \in D$, $\eta \in J(\xi)$ and $0 \leq t <+\infty$

\[
\int_0^t\,(\|F_s\eta\|^2+\|F^*_s\eta\|^2)\,d\mu_{\xi}(s) < +\infty
\]

The thus extended stochastic integral has the usual linearity properties, and the maps $(s,t) \rightarrow \int_s^t\,dM_z\,F_z$ and $(s,t) \rightarrow \int_s^t\,F_z\,dM_z$  are additive, adapted processes, strongly continuous on $D$.

Suppose that $I$ is a set of finite cardinality  and let $\{M_{\alpha}\,/\,\alpha \in I\}$ be a set of integrators of scalar type.  Consider the quantum stochastic differential equation

\[
X(t)=X_0+\int_0^t\,\sum_{\alpha \in I}\,dM_{\alpha}(s)\,F_{\alpha}(s)X(s)G_{\alpha}(s)
\]

or in differential form

\[
dX(t)=\sum_{\alpha \in I}\,dM_{\alpha}(t)\,F_{\alpha}(t)X(t)G_{\alpha}(t),\,\,\,X(0)=X_0,\,\,t \geq 0
\]

where $X_0 \in A_{0]}$ and the coefficients $F_{\alpha},G_{\alpha}$ are locally bounded adapted processes leaving the domain $D$ invariant. If for all $\alpha \in I$, for all adapted processes $P$ integrable with respect to $M_{\alpha }$, and for all continuous functions $u,v$ on $[0,+\infty)$ satisfying $u(s) \leq s, v(s) \leq s$ for all $s \in [0,+\infty)$, the family of operators on $H$ $\{s \rightarrow F_{\alpha}(u(s))\,P(s)\,G_{\alpha}(v(s))\}$ is an adapted process integrable with respect to $M_{\alpha }$, then the above quantum stochastic differential equation has a unique locally bounded solution $X$ which is strongly continuous on $D$.

The above result can easily be extended to equations of the form

\[
dX(t)=\sum_{\alpha \in I}\,dM_{\alpha}(t)\,F_{\alpha}(t)(w(t)X(t)+z(t))G_{\alpha}(t),\,\,\,X(0)=X_0,\,\,t \geq 0
\]

where $w,z$ are locally bounded adapted processes leaving the domain $D$ invariant.

Following \cite{b} we restrict, in what follows, the term \textit{process} to processes leaving the domain $D$ invariant, and we denote the *-algebra of all processes by $W$.

If $M=\{M_{\alpha}\,/\,\alpha \in I \}$ is a self-adjoint family of regular integrator processes (i.e $M_{\alpha} \in M$ implies $(M_{\alpha})^* \in M$) then for all $s,t \in [0,+\infty)$ with $s < t$, for all $\alpha, \beta \in I$, for all $\xi \in D$, and for all adapted processes $F$, the \textit{Meyer Bracket} or \textit{mutual quadratic variation } of $M_{\alpha}$ and $M_{\beta}$, defined by

\[
[[M_{\beta},M_{\alpha}]](s,t)=\lim_{|\Pi| \rightarrow 0}\, \sum\,M_{\beta}(t_{k-1},t_k)\,M_{\alpha}(t_{k-1},t_k)\,F_s\xi
\]

where $\Pi$ is a partition of $[s,t]$,  exists in norm and defines an additive adapted process satisfying

\begin{eqnarray*}
&M_{\beta}(s,t)\,M_{\alpha}(s,t)\,F_s\xi= \{\int_s^t\,dM_{\beta}(r)\,M_{\alpha}(s,r)+\int_s^t\,dM_{\alpha}(r)\,M_{\beta}(s,r)+&\\
&[[M_{\beta},M_{\alpha}]](s,t)\}\,F_s\xi&
\end{eqnarray*}

Assuming, for each pair $(\alpha,\beta) \in I \times I$, the existence of a family $\{c^{\gamma}_{\alpha \beta}\,/\,\gamma \in I \}$ of \textit{structure processes} such that for each $s,t \in [0,+\infty)$ with $s < t$

\[
[[M_{\beta},M_{\alpha}]](s,t)=\sum_{\gamma \in I}\,\int_s^t\,c^{\gamma}_{\alpha \beta} (r)\,dM_{\gamma}(r)
\]

and defining the \textit{differential} of an additive process $M$ by

\[
dM(t)=M(t,t+dt)
\]

we obtain

\begin{eqnarray}
d(M_{\beta}\,M_{\alpha})(t)=dM_{\beta}(t)\,M_{\alpha}(t)+M_{\beta}(t)\,dM_{\alpha}(t)+dM_{\beta}(t)\,dM_{\alpha}(t)
\end{eqnarray}

where the last product on the right is computed with the use of the \textit{It\^{o} table}

\begin{eqnarray}
dM_{\beta}(t)\,dM_{\alpha}(t)=\sum_{\gamma \in I}\,c^{\gamma}_{\alpha \beta} (t)\,dM_{\gamma}(t)
\end{eqnarray}

Assuming further that the $M_{\alpha}\,'s$ satisfy a \textit{ $\rho$-commutation relation} i.e that for each $\alpha \in I$ there exists an automorphism $\rho_{\alpha}$ of $W$ mapping adapted processes into adapted processes, and such that

\begin{eqnarray}
 \rho_{\alpha}^2=id
\end{eqnarray}

\medskip

where $id$ denotes the identity map,

 and for every $\xi \in D$, $s < t$, and adapted processes $F$, $F_s\xi \in dom(M_{\alpha}(s,t))$, $M_{\alpha}(s,t)\xi \in dom(\rho_{\alpha}(F_s))$,

\begin{eqnarray}
M_{\alpha}(s,t)F_s\xi=\rho_{\alpha}(F_s)M_{\alpha}(s,t)\xi  
\end{eqnarray}

and

\begin{eqnarray}
F_sM_{\alpha}(s,t)\xi=M_{\alpha}(s,t)\rho_{\alpha}(F_s)\xi
\end{eqnarray}

i.e stochastic processes commute with the stochastic differentials of the integrators,  we can extend (2.1) to processes $X=\{X(t)\,/\,t \geq 0\}$ and  $Y=\{Y(t)\,/\,t \geq 0\}$ of the form

\[
X(t)=\sum_{\alpha \in I}\,\int_s^t\,dM_{\alpha}(z)\,H_{\alpha}(z)\,,\,\,Y(t)=\sum_{\alpha \in I}\,\int_s^t\,dM_{\alpha}(z)\,K_{\alpha}(z)
\]

where $H_{\alpha}, K_{\alpha}$ are for each  $\alpha \in I$ strongly continuous adapted processes. 

We thus have

\begin{eqnarray}
d(X \, Y)(t)=dX(t)\,Y(t)+X(t)\,dY(t)+dX(t)\,dY(t)
\end{eqnarray}

where

\[
dX(t)=\sum_{\alpha \in I}\,dM_{\alpha}(t)\,H_{\alpha}(t),\,\,dY(t)=\sum_{\alpha\in I}\,dM_{\alpha}(t)\,K_{\alpha}(t)
\]
 
$dX(t)\,dY(t)$ is computed with the use of the It\^{o} table (2.2), and (2.6) is understood weakly on $D$ i.e for all $\xi,\eta \in D$ 

\[
<d(X \, Y)(t)\xi,\eta>=<[dX(t)\,Y(t)+X(t)\,dY(t)+dX(t)\,dY(t)]\xi,\eta>
\]

\medskip
 \section{\textbf{ Quantum  feedback control }}
\medskip

Within the framework of the AFQ calculus described in the previous section, we consider an operator process $X=(X(t))_{t\geq0}$, defined on a complex separable Hilbert space $H$ containing a total invariant subset $D$ , with evolution described by a quantum stochastic differential equation of the form

\begin{eqnarray}
&dX(t)=d\tau(t)(FX+Gu+L)(t)+\sum_{a\in I}dM_a(t)F_a(t)(wX+z)(t)&\nonumber\\
\\
&X(0)=C,\, 0\leq t \leq T < +\infty&\nonumber
\end{eqnarray}

or of the form

\begin{eqnarray}
&dX(t)=-\{d\tau(t)(FX+Gu+L)(t)+\sum_{a\in I}dM_a(t)F_a(t)(wX+z)(t)\}&\nonumber\\
&&\\
&X(T)=C,\, 0\leq t \leq T < +\infty&\nonumber
\end{eqnarray}

where $I$ is a finite set, $M=\{M_a \,/\,a\in I\}$ is a self-adjoint  family of  regular integrators of scalar type  satisfying a $\rho$-commutation relation with It\^{o} multiplication rules

\begin{eqnarray}
dM_a(t)\,dM_b(t)&=&\sum_{l \in I}c_{ab}^l (t)\,dM_l (t)
\end{eqnarray}

and

\begin{eqnarray}
dM_a (t)\,d\tau (t)&=&d\tau (t)\,dM_a (t)=0
\end{eqnarray}

\medskip

where the $c_{ab}^l$ 's are \textit{structure processes}, $\tau $ is a real-valued measure on $[0,+\infty )$, absolutely continuous with respect to Lebesgue measure,    $dM_a \ne d\tau $ for all $a \in I$, and  the coefficient processes $F, G, u, L, w, z$ and $F_a$ for all $a \in I$, are as in Section 2. We assume also that $C$ is a bounded operator on $H$.  As shown in \cite{b}, (3.1) and (3.2) admit unique locally bounded solutions.

Under extra assumptions on the coefficients, e.g if $F,G,L,w,z$ are real-valued functions and the $F_{\alpha}$'s are complex-valued functions with conjugate $\overline{F_{\alpha}}=F_{\alpha^*}$ where $a^*$ is defined by $(M_a)^*=M_{a^*}$, then  the $X(t)\,'s$ correspond to classical quantum mechanical observables.

 As in \cite{h} for the classical case we associate with (3.1) and (3.2)   respectively the quadratic performance criteria (3.5) and (3.6)  of the following

\begin{definition}

For $\xi \in D$, $0\leq T < +\infty$, $Q,\,R,\,m,\,\eta$  locally bounded, strongly continuous, adapted processes such that $R$ has an inverse $R^{-1}$ with the same properties, for $Q_T,\,m_T,\,Q_0,\,m_0$   bounded operators on $H$ with $R(t)\geq 0$, $Q(t)\geq 0$, $R^{-1}(t)>0$ for all $t\in [0,T]$, and for $Q_T\geq 0$, $Q_0\geq 0$  define

\begin{eqnarray}
&\tilde J_{\xi,T}(u)=\int_0^T\,d\tau(t)(\langle X(t)\xi,Q(t)X(t)\xi\rangle + \langle u(t)\xi,R(t)u(t)\xi\rangle& \\
&+2\langle m(t)X(t)\xi,\xi\rangle+2\langle \eta(t)u(t)\xi,\xi\rangle)&\nonumber\\
&+\langle Q_T\,X(T)\xi,X(T)\xi\rangle +2\langle m_T\,X(T)\xi,\xi\rangle\nonumber&
\end{eqnarray}

and
 
\begin{eqnarray}
&\tilde J_{\xi,0}(u)=\int_0^T\,d\tau(t)(\langle X(t)\xi,Q(t)X(t)\xi\rangle + \langle u(t)\xi,R(t)u(t)\xi\rangle &\\
&+2\langle m(t)X(t)\xi,\xi\rangle+2\langle \eta(t)u(t)\xi,\xi\rangle)&\nonumber\\
&+\langle Q_0\,X(0)\xi,X(0)\xi\rangle +2\langle m_0\,X(0)\xi,\xi\rangle\nonumber&
\end{eqnarray}

\end{definition}

We view $u$ as a control process and we consider the problem of choosing it so as to minimize $\tilde J_{\xi,T}(u)$ (resp. $\tilde J_{\xi,0}(u)$), thus controlling the evolution of the process $X$ described by (3.1), (resp. (3.2)).

In the classical  case, i.e in the case of noise described by classical, scalar or vector-valued Brownian motion, the problem is well-studied (see e.g \cite{h}). In the  quantum case fundamental work on the subject of control and filtering  has been done in \cite{e,f,g} and also in \cite{a,j,k,l,l2}.

\medskip

\begin{theorem}
Let $T>0$ be a finite time, let $X=\{X(t)\,/\,t \geq  0\}$ be a locally bounded adapted process with evolution described by (3.1) (resp.(3.2)) and with performance criterion (3.5) (resp.(3.6)), and suppose that there exists a self-adjoint, locally bounded process $ \Pi = \{ \Pi (t)\,/\,t \geq 0\}$  satisfying, weakly on the invariant domain $D$, the generalized  stochastic Riccati differential equation

\begin{eqnarray}
&d\tau(t)(F^*\Pi +\Pi F+Q-\Pi GR^{-1}G^*\Pi )(t)+[(\sum_{a\in I}dM_a\,F_aw)^* \,\Pi + &\nonumber\\
&\Pi \,\sum_{a\in I}dM_a \,F_a w \pm (\sum_{a\in I}dM_a \,F_a w)^* \,\Pi \,(\sum_{a\in I}dM_a \,F_a w)](t) \pm & \nonumber\\
&(\sum_{a\in I}dM_a  \,F_aw \pm id)^*(t)\,d\Pi (t)\,(\sum_{a\in I} dM_a \,F_a w \pm id)(t)=0&\nonumber\\
&&\\
&\Pi(T)=Q_T(\mbox{resp. }\,\Pi(0)=Q_0), 0\leq t \leq T &\nonumber
\end{eqnarray}

and a locally bounded adapted process $r=\{r(t)\,/\,t \geq 0\}$ satisfying the stochastic differential equation

\begin{eqnarray}
&d\tau(t)(F^*r-\Pi GR^{-1}G^*r+\Pi L+m^*-\Pi GR^{-1}\eta^* )(t)+[(\sum_{a\in  I}dM_a\, F_a w)^* r+&\nonumber\\
&\Pi \,\sum_{a\in I}dM_a \, F_a z + d\Pi \, \sum_{a\in I}dM_a \, F_a z \pm  (\sum_{a\in I}dM_a \, F_a w)^* \,\Pi \,\sum_{a\in I}dM_a \, F_a z \pm &\nonumber\\
&(\sum_{a\in I}dM_a\, F_a w)^*\,d\Pi \,\sum_{a\in I}dM_a \, F_a z ](t)+  [(\sum_{a\in I}dM_a \, F_a w \pm id)^* dr](t) =0&\nonumber\\
&&\\
&r(T)=m_T^* (\mbox{resp.  }\,r(0)= m_0^* ) , 0\leq t \leq T,&\nonumber
\end{eqnarray}

where $id$ denotes the identity operator on $H$, the plus sign in $\pm$ in (3.7) and (3.8) is associated with (3.1) and (3.5), and the minus with (3.2) and (3.6). Then the performance criterion $\tilde J_{\xi,T}(u)$ (resp. $\tilde J_{\xi,0}(u)$) appearing in (3.5) (resp.  (3.6)) is minimized by the feedback control process

\begin{eqnarray*}
u=-R^{-1}(G^*(\Pi X+r)+\eta^*).
\end{eqnarray*}

Note: For $w=0$ and $z=id$ we obtain the solution to the quantum analogue of the "linear regulator" problem of classical control theory.

\end{theorem}

\begin{proof} We will give the proof for (3.1) and (3.5). The proof for (3.2) and (3.6) is similar. So let $u=\Lambda X+\lambda+ \mu$ where $\Lambda,\,\lambda$ are fixed processes to be chosen later and $\mu$ is the new control. We will choose $\Lambda,\,\lambda$ so that the minimizing new control $\mu$ is identically $0$. Replacing $u$ by $\Lambda X+\lambda+ \mu$ in (3.1) we obtain

\begin{eqnarray*}
&dX(t)=d\tau (t)(FX+G\Lambda X+G\lambda +G\mu +L)(t)+\sum_{a\in I}dM_a(t)F_a(t)(wX+z)(t)&\nonumber\\
&&\\
&X(0)=C,\, 0\leq t \leq T < +\infty &\nonumber
\end{eqnarray*}

Let $Y$ be the solution of the above equation corresponding to $\mu=0$, i.e

\begin{eqnarray}
&dY(t)=d\tau (t)(FY+G\Lambda Y+G\lambda  +L)(t)+\sum_{a\in I}dM_a(t)F_a(t)(wY+z)(t)&\nonumber\\
&&\nonumber\\
&Y(0)=C,\, 0\leq t \leq T < +\infty &\nonumber
\end{eqnarray}

with corresponding control $u_0=\Lambda Y+\lambda$. Letting $\hat X=X-Y$ we obtain

\begin{eqnarray}
&d\hat X(t)=d\tau (t)(F \hat X+G\Lambda \hat X +G\mu)(t)+\sum_{a\in I}dM_a(t)F_a(t)(w \hat X)(t)&\nonumber\\
&&\nonumber\\
&\hat X(0)=0,\, 0\leq t \leq T < +\infty &\nonumber
\end{eqnarray}

and using $u=\Lambda \hat X+ u_0 +\mu$ (3.5) becomes

\begin{eqnarray}
&\tilde J_{\xi,T}(u)=\tilde J_{\xi,T}(u_0)+ \int_0^T\,d\tau (t)(\langle \xi,[\hat X^*Q \hat X +(\Lambda \hat X+\mu)^*R(\Lambda \hat X + \mu)](t)\xi \rangle& \nonumber\\
&+\langle \xi,\hat X^*(T)Q_T\hat X(T)\xi \rangle +2\Re\,K&
\end{eqnarray}

where $\Re\,K$ denotes the real part of $K$ and

\begin{eqnarray}
&K=\int_0^T\,d\tau (t)\langle \xi,[\hat X^*QY +(\Lambda \hat X+\mu)^*R(\Lambda  Y + \lambda) +\hat X^*m^*+&\\
&(\Lambda \hat X+\mu)^*\eta^*](t)\xi \rangle + \langle \xi,[\hat X^*(T)m_T^*+
\hat X^*(T)Q_TY(T)]\xi \rangle \nonumber&
\end{eqnarray}

We will show that if 

\begin{eqnarray}
\Lambda=-R^{-1}G^*\Pi,\, \lambda =-R^{-1}(G^*r+\eta^*)
\end{eqnarray}

 then $K=0$. In view of (3.9) we will then conclude that (3.5) is minimized by $\mu=0$. In so doing, let $p(t)=r(t)+\Pi(t)Y(t)$. Then

\begin{eqnarray}
 \langle \xi,[\hat X^*(T)m_T^*+\hat X^*(T)Q_TY(T)]\xi \rangle =\int_0^T\,d\langle \xi,(\hat X^*p)(t)\xi \rangle 
\end{eqnarray}

Using

\begin{eqnarray}
&d\hat X^*(t)=d\tau (t)( \hat X ^* F^* + \hat X ^*{\Lambda}^* G^* +{\mu}^* G^*)(t)+\sum_{a\in I}dM_{a^*}(t){\rho}_{a^*} ( \hat X^* w^* F_a^*)(t)&\nonumber\\
&&\nonumber\\
&\hat X^*(0)=0,\, 0\leq t \leq T < +\infty &\nonumber
\end{eqnarray}

where $a^*$ is defined by $(M_a)^*=M_{a^*}$, and (2.6) to compute the right hand side of (3.12),  (3.10) becomes

\begin{eqnarray*}
&K=\int_0^T\,<\xi, \{\hat X^* [d\tau ((F^* \Pi +\Lambda^* G^* \Pi +\Pi (F+G \Lambda)+Q+\Lambda^* R \Lambda)+&\\
&\sum_{a \in I} w^* F_a^* \,dM_{a^*} \, \Pi + \Pi  \, \sum_{a \in I}dM_a \,F_a w +d\Pi \, \sum_{a \in I}dM_a \,F_a w+&\\
 &\sum_{a \in I}w^* F_a^*\,dM_{a^*}\,d\Pi+ \sum_{a \in I}w^* F_a^*\,dM_{a^*}\,\Pi\,\sum_{a \in I}dM_a \,F_a w+&\\
& \sum_{a \in I}w^* F_a^*\,dM_{a^*}\,d\Pi \,\sum_{a \in I}dM_a \,F_a w+d\Pi)Y+(d\tau(F^* r+\Lambda^* G^* r+&\\
&\Pi(G \lambda+L)+\Lambda^* R \lambda+m^*+\Lambda^* \eta^* )+\sum_{a \in I}w^* F_a^* \,
dM_{a^*}\,r+\Pi \sum_{a \in I}dM_a \,F_a z+&\\
&d\Pi \,\sum_{a \in I}dM_a \,F_a z +\sum_{a \in I}w^* F_a^* \,dM_{a^*}\,dr+\sum_{a \in I}w^* F_a^*\,
dM_{a^*}\,\Pi\,\sum_{a \in I}dM_a \,F_a z +&\\
&\sum_{a \in I}w^* F_a^*\,
dM_{a^*}\,d\Pi \,\sum_{a \in I}dM_a \,F_a z+dr)]+\mu^* [(G^* \Pi+R^* \Lambda )Y+&\\
&(G^*r+R \lambda+\eta^*)]\,d\tau \}(t)\xi >&
\end{eqnarray*}

Replacing in the above $\Lambda$ and $\lambda$ by $-R^{-1}G^*\Pi$ and $-R^{-1}(G^*r+\eta^*)$ respectively we see that the coefficients of $Y$ are, in view of (3.7), equal to zero. The same is true, by (3.8), for the constant terms. Thus $K=0$.

\end{proof}
 
\medskip

\begin{definition} Following \cite{b},
the family of integrators of scalar type \break $\{M_0,M_a\,/\,a\in I\}$ where $dM_0=d\tau$, is said to be linearly independent if the equality

\[
 \int_0^t\,d\tau(s)\,G(s)+\sum_{a\in I}\,dM_a(s)\,G_a(s)=0
\]

 for all families $\{G,G_a\,/\,a\in I\}$ of adapted processes and all $t\geq 0$, implies that $G=G_a=0$ for all $a\in I$.
\end{definition}

\begin{proposition}
 If the family of integrators of scalar type appearing in (3.1) and (3.2) is linearly independent then the Riccati equation (3.7) can be put in the form 

\[
d\Pi (t)= d\tau (t)\,A(t,\Pi (t))+ \sum_{a\in I}\,dM_a(t)\,B_a(t,\Pi (t))
\]

 where $A(t)=A(t,\Pi (t))$ and $B_a(t)=B_a(t,\Pi (t))$ can be described as the solutions of the operator equations

\begin{eqnarray}
&(F^*\Pi +\Pi F+Q-\Pi GR^{-1}G^*\Pi \pm A \pm \sum_{a,b\in I}[c_0(a,b)\rho _b(\rho _a(w^* F_{j(a)}^*)\Pi )F_b w\pm &\nonumber\\
 &c_0(a,b)\rho _b(\rho _a(w^*F_{j(a)}^*))B_b\pm c_0(a,b)\rho _b (B_a)F_b w+\sum_{\gamma,\epsilon \in I}c_{\epsilon}(a,b)c_0(\epsilon,\gamma)&\nonumber\\
&\rho _{\gamma}(\rho _b(\rho _a(w^*F_{j(a)}^*))B_b)F_{\gamma} w] )(t)=0,\,\,\, \forall t \in [0,T]&\nonumber  
\end{eqnarray}

and for all $J \in I$  

\begin{eqnarray}
&(\rho _J(w^*F_{j(J)}^*)\Pi +\rho_J (\Pi) F_J w \pm B_J + \sum_{a,b\in I}[c_J (a,b)\rho _b (B_a) F_b w +&\nonumber\\
& c_J (a,b) \rho_b (\rho_a(w^*F_{j(a)}^*))B_b \pm c_J(a,b)\rho_b (\rho_a (w^* F_{j(a)}^*)\Pi )F_b W \pm \sum_{\gamma,\epsilon \in I} &\nonumber\\
&\rho_J(c_{\epsilon}(a,b))c_J(a,\gamma)\rho _{\gamma}(\rho _b(\rho _a(w^*F_{j(a)}^*))B_b)F_{\gamma} w])(t)=0,\,\,\, \forall t \in [0,T]&\nonumber
\end{eqnarray}

while (3.8) can be put in the form

\[
dr(t)=d\tau (t)\,C(t,r(t))+\sum_{a\in I}\,dM_a(t)\,D_a(t,r(t))
\]

 where $C(t)=C(t,r(t))$ and $D_a(t)=D_a(t,r(t))$ can be described as the solutions of the operator equations

\begin{eqnarray}
&(F^*r -\Pi GR^{-1}G^* r+\Pi L+m^* -\Pi GR^{-1} {\eta}^*  \pm C+ \sum_{a,b\in I}[c_0(a,b)\rho _b (B_a)F_b z & \nonumber\\
 &\pm c_0(a,b)\rho _b(\rho _a(w^* F_{j(a)}^*)\Pi )F_b z+
 c_0(a,b)\rho _b(\rho _a(w^*F_{j(a)}^*))D_b& \nonumber\\
 &\pm \sum_{\gamma,\epsilon \in I}c_{\epsilon}(a,b)c_0(\epsilon,\gamma)])(t)=0, \forall t \in [0,T]&\nonumber  
\end{eqnarray}

and for all $J \in I$  

\begin{eqnarray}
&(\rho _J(w^*F_{j(J)}^*)r +\rho_J (\Pi) F_J z \pm D_J + \sum_{a,b\in I}[c_J (a,b)\rho _b (B_a) F_b z +& \nonumber\\
 &\rho_b (\rho_a(w^*F_{j(a)}^*)\Pi)F_b z +\rho_b (\rho_a (w^* F_{j(a)}^*) )D_b ) \pm \sum_{\gamma,\epsilon \in I} \rho_J(c_{\epsilon}(a,b))c_J(\epsilon,\gamma)&\nonumber\\
&\rho _{\gamma}(\rho _b(\rho _a(w^*F_{j(a)}^*)))\rho _{\gamma} (B_b)F_{\gamma}z])(t)=0, \forall t \in [0,T]&\nonumber
\end{eqnarray}

Here the adapted processes $c_0(a,b),c_{\epsilon}(a,b),a,b\in I$ are defined for all $t\geq 0$ by

\begin{eqnarray}
dM_a(t)\,dM_b(t)=d\tau (t) \,c_0(a,b)(t)+ \sum_{\epsilon \in I}\,dM_{\epsilon}(t)\,c_{\epsilon}(a,b)(t)
\end{eqnarray}

 $\rho_a$ is, for each $a\in I$, the corresponding commutation homomorphism, and the mapping $j:I\rightarrow I$ is defined by $j(a^*)=a$.

\end{proposition}

\begin{proof} We will only give the proof for $ \Pi = \{ \Pi (t)\,/\,t \geq 0\}$ . The proof for 
\break  $r=\{r(t)\,/\,t \geq 0\}$ is similar. Substituting  $d\Pi = d\tau \,A+ \sum_{a \in I}\,dM_a \, B_a$ in (3.7) we obtain

\begin{eqnarray*}
&\{d\tau\,(F^*\Pi +\Pi F+Q-\Pi GR^{-1}G^*\Pi)+\sum_{a \in I}\,dM_{a^*} \,\rho_{a^*}\,(w^*F^*_a)\Pi+&\\
&\Pi \sum_{a \in I}\,dM_a \, F_a w \pm \sum_{a \in I}\,dM_{a^*} \,\rho_{a^*} \,(w^*F^*_a)\Pi \sum_{b \in I}\,dM_b \, F_b w + &\\
&\sum_{a \in I}\,dM_{a^*} \,\rho_{a^*} \,(w^*F^*_a)(d \tau \,A + \sum_{b \in I}\,dM_b \,B_b)+&\\
&(d \tau \,A+\sum_{a \in I}\,dM_a \,B_a)\sum_{b \in I}\,dM_b \,F_b w &\\
&\pm \sum_{a \in I}\,dM_{a^*} \,\rho_{a^*}\,(w^*F^*_a)(d\tau \,A+ \sum_{b \in I}\,dM_b \, B_b)\sum_{\gamma \in I}\,dM_{\gamma} \, F_{\gamma} w &\\
&\pm
d\tau \,A \pm \sum_{a \in I}\,dM_a \, B_a \}(t)=0&
\end{eqnarray*}

Making use of (3.4), of the $\rho$-commutation relations,  of (3.13), and of

\[
dM_a(t)\,dM_b(t)\,dM_{\gamma}(t)=\sum_{\epsilon, \delta \in I}\,dM_{\epsilon}(t)\,c_{\epsilon}(\delta,\gamma)(t)\rho_{\gamma}(c_{\delta}(a,b)(t))
\]

we obtain after renaiming the indices,

\begin{eqnarray*}
&\{d\tau\,(F^*\Pi +\Pi F+Q-\Pi GR^{-1}G^*\Pi \pm A \pm \sum_{a,b \in I} [c_0(a,b)\rho_b(\rho_a(w^*F^*_{j(a)})\Pi)F_bw &\\
&\pm c_0(a,b)\rho_b(\rho_a(w^*F^*_{j(a)}))B_b \pm c_0(a,b)\rho_b(B_a)F_bw +&\\
&\sum_{\gamma,\epsilon}c_{\epsilon}(a,b)c_0(\epsilon,\gamma) \rho_{\gamma} (\rho_b (\rho_a(w^*F^*_{j(a)}))B_b)F_{\gamma}w])+&\\
& \sum_{J \in I}\,dM_J(\rho_J(w^* F^*_{j(J)})\Pi+\rho_J(\Pi)F_Jw \pm B_J +\sum_{a,b \in I} [c_J(a,b) \rho_b(B_a)F_bw +&\\
&c_J(a,b) \rho_b(\rho_a(w^*F^*_{j(a)}))B_b \pm c_J(a,b)\rho_b(\rho_a(w^*F^*_{j(a)})\Pi)F_bw \pm  &\\
&\sum_{\epsilon, \gamma \in I} \rho_J ( c_{\epsilon}(a,b))c_J(a,\gamma)\rho_{\gamma} (\rho_{b} (\rho_a(w^*F^*_{j(a)}))B_b)F_{\gamma}w])\}(t)=0&
\end{eqnarray*}

from which the result follows by the linear independence assumption.

\end{proof}

\medskip

\section{\textbf{Quantum stochastic Riccati equations}}

\begin{definition}
Following  \cite{b} let $M,M_1,M_2$ be integrator processes of scalar type such that $dM_1=dM$ and $dM_2=dM^*$. If the Meyer bracket $[[M_b,M_a]]$ exists for all  $a,b\in \{1,2\}$ and is a complex-valued nonatomic measure, then the pair $(M_1,M_2)$ is called a Levy pair.

 A Levy pair $(M_1,M_2)$ is said to be of Boson type if $[M_b(I),M_a(J)]=0$ for all  $a,b\in \{1,2\}$ and $I,J \subset [0,+\infty )$ with $I \cap J= \varnothing$, and it is said to be of Fermion type if $\{M_b(I),M_a(J)\}=0$ where $[x,y]=xy-yx$ and $\{x,y\}=xy+yx$. 

For a Boson (resp. Fermion) type Levy pair, $\rho_1=\rho_2=id$ \break (resp. $\rho_1=\rho_2=-id$) where $\rho_1,\rho_2$ are the commutation automorphisms corresponding to $M_1,M_2$ respectively. 

For a Levy pair $(M_1,M_2)$, $dM_b^*(t)\,dM_a(t)=\sigma_{ba}(t)\,dt$, for all  
$a,b\in \{1,2\}$, where the matrix valued function $t\rightarrow (\sigma_{ba}(t))_{a,b\in \{1,2\}}$ is positive definite in the sense that for all complex-valued  continuous functions $f$, $(f\,\,f)\cdot \sigma \cdot (f\,\,f)^t \geq 0$. 

Let $M_0$ be an integrator of scalar type such that $dM_0 (t)=dt$ where $dt$ is the usual time differential. If $M_0,M_1,M_2$ are linearly independent in the sense of Definition 3, then $(M_1,M_2)$ is called a linearly independent Levy pair.

\end{definition}

If $\{M_a\,/\,a \in I\}=\{M_1,M_2\}$ is a linearly independent Levy pair (e.g a Boson or a Fermion Levy pair) then Theorem 1 includes the solution to the control problem of stochastic evolutions driven by quantum Brownian motion in terms of the solution of a stochastic Riccati equation to be studied in more detail in this section.

 In that case  the Riccati equation (3.7) reduces to

\begin{eqnarray}
&d\Pi (t)=dt((\mp F+\sigma _{11} \rho _2 (F_2 w) \rho _1 \rho _2 (F_1 w)+\sigma _{12} 
\rho _2 (F_2 w)F_2 w+&\nonumber\\
&\sigma _{22} 
\rho _1 (F_1 w)\rho _2 \rho _1 (F_2 w)+\sigma _{21 }\rho _1 (F_1 w) F_1 w)^* \Pi + (\mp \Pi F +\sigma_{11} 
 \rho _1 \rho _2 (\Pi )\rho _1 (F_2 w)F_1 w\nonumber\\
&+\sigma _{12} \Pi \rho _2 (F_2 w)F_2 w +\sigma _{22} \rho _2 \rho _1 (\Pi )\rho _2(F_1 w)F_2 w+\sigma _{21} \Pi  \rho _1 (F_1 w)F_1 w)+& \nonumber\\
&\sigma _{11}\rho _1 \rho _2( w^* F_1^* )\rho _1(\Pi )F_1 w
+\sigma _{12} w^* F_1^* \rho _2 (\Pi )  F_2 w  + \sigma _{22} 
\rho _2 \rho _1 (w^* F_2^* )\rho _2(\Pi )F_2 w& \nonumber\\
&+\sigma_{21} w^* F_2^* \rho_1  ( \Pi) F_1 w \mp Q \pm \Pi GR^{-1}G^* \Pi )(t) \mp  dM_1(t)(\rho _1(w^* F_2^* )\Pi  & \nonumber\\
& +\rho _1 (\Pi) F_1 w)(t) \mp  dM_2(t)(\rho _2(w^* F_1^* )\Pi+  \rho _2 (\Pi ) F_2 w)(t)& \nonumber\\
&&\\
&\Pi (T)=Q_T\,\,(\mbox{resp.}\,\Pi (0)=Q_0), 0\leq t \leq T& \nonumber
\end{eqnarray}   

where the plus (resp. minus) sign in $\pm$ (resp. in $\mp$) is associated with (3.1) and (3.5), and the minus (resp. plus) with (3.2) and (3.6).

\begin{theorem}
The Riccati equation  (4.1) admits a unique, adapted, strongly continuous, positive, locally bounded, solution $\Pi =(\Pi (t))_{0 \leq t \leq T}$ defined weakly on the invariant domain $D$.
\end{theorem}

\begin{proof}
We first consider the case corresponding to (3.1) and (3.5). Equation (4.1) can be written as

\begin{eqnarray}
&d\Pi (t)=[dt( F+\sigma _{11} \rho _2 (F_2 w) \rho _1 \rho _2 (F_1 w)+\sigma _{12} 
\rho _2 (F_2 w)F_2 w+&\nonumber\\ 
 & \sigma _{22} \rho _1 (F_1 w)\rho _2 \rho _1 (F_2 w)
+\sigma _{21 }\rho _1 (F_1 w) F_1 w- GR^{-1}G^* \Pi )(t)+dM_1(t)( F_1 w)(t)&\nonumber\\
& + dM_2(t)( F_2 w )(t)]^*\Pi (t)+[dt ( \Pi F +\sigma_{11} \rho _1 \rho _2 (\Pi )\rho _1 (F_2 w)F_1 w+&\nonumber\\
&\sigma _{12} \Pi \rho _2 (F_2 w)F_2 w + \sigma _{22} \rho _2 \rho _1 (\Pi )\rho _2(F_1 w)F_2 w +\sigma _{21} \Pi  \rho _1 (F_1 w)F_1 w-&\nonumber\\
&\Pi GR^{-1}G^* \Pi )(t)+dM_1(t)(\rho _1(\Pi ) F_1 w)(t)
 + dM_2(t)(\rho _2(\Pi ) F_2 w )(t)]&\nonumber\\
&+dt[\sigma _{11}\rho _1 \rho _2( w^* F_1^* )\rho _1(\Pi )F_1 w 
+\sigma _{12} w^* F_1^* \rho _2 (\Pi )  F_2 w &\nonumber\\
 & + \sigma _{22} \rho _2 \rho _1 (w^* F_2^* )\rho _2(\Pi )F_2 w+\sigma_{21} w^* F_2^* \rho_1  ( \Pi) F_1 w ](t) +dt( Q + \Pi GR^{-1}G^* \Pi )(t)&\nonumber\\
&&\\
&\Pi (0)=Q_0\,\,, 0\leq t \leq T &\nonumber
\end{eqnarray}

Using (2.6), the identities  $dt\,dM_1=dt\,dM_2=dM_1\,dt=dM_2\,dt=0$, $dM_1\,dM_2=dM_2^*\,dM_2=\sigma _{22} dt$, $ dM_2\,dM_1=dM_1^*\,dM_1=\sigma _{11} dt$,  $dM_2\,dM_2=dM_1^*\,dM_2=\sigma _{12} dt$, $dM_1\,dM_1=dM_2^*\,dM_1=\sigma _{21} dt $, and the fact that if $\{ \lambda (t,s)\,/\,t \mbox{(resp.}s) \geq 0 \}$ is for each $s$ (resp. $t$) a process then 

\begin{eqnarray*}
d\left(\int_0^t\,ds\,\lambda (t,s)\right)=dt\,\lambda (t,t)+\int_0^t\,ds\, \,\,d\lambda (t,s)   
\end{eqnarray*}

we can prove by taking the time differential of the right hand side and showing that it satisfies (4.2), that weakly on $D$

\begin{eqnarray*}
\Pi (t)=K(t,0)Q_0K(t,0)^*+\int_0^t\,ds\,K(t,s)(Q + \Pi GR^{-1}G^* \Pi )(s)K(t,s)^*
\end{eqnarray*}

where

\begin{eqnarray*}
&dK(t,s)=[dt( F+\sigma _{11} \rho _2 (F_2 w) \rho _1 \rho _2 (F_1 w)+\sigma _{12} 
\rho _2 (F_2 w)F_2 w+&\\
&\sigma _{22} \rho _1 (F_1 w)\rho _2 \rho _1 (F_2 w)+\sigma _{21 }\rho _1 (F_1 w) F_1 w- GR^{-1}G^* \Pi )(t)&\\
&+dM_1(t)( F_1 w)(t)\ + dM_2(t)( F_2 w )(t)]^* K(t,s)&\\
&&\\
&K(s,s)=id,\,\,s\leq t \leq T &
\end{eqnarray*}

Let the sequence $\{ \Pi _n \} _{n=1}^{+\infty}$ of locally bounded self-adjoint processes, be defined by the iteration scheme

\[
\Pi _1 (t) = Q_0
\]

and for $n \geq 1$

\[
\Pi _{n+1} (t)=K_n(t,0)\,Q_0\,K_n(t,0)^*+\int_0^t\,ds\,K_n(t,s)(Q + \Pi _n GR^{-1}G^* \Pi _n )(s)K_n(t,s)^*
\]

where $K_n(t,s)$ is the unique locally bounded solution of

\begin{eqnarray*}
&dK_n(t,s)=[dt( F+\sigma _{11} \rho _2 (F_2 w) \rho _1 \rho _2 (F_1 w)+\sigma _{12} 
\rho _2 (F_2 w)F_2 w&\\
&+\sigma _{22} 
\rho _1 (F_1 w)\rho _2 \rho _1 (F_2 w)
+\sigma _{21 }\rho _1 (F_1 w) F_1 w- GR^{-1}G^* \Pi )(t)&\\
&+dM_1(t)( F_1 w)(t)\ + dM_2(t)( F_2 w )(t)]^* K_n(t,s)&\\
&&\\
&K_n(s,s)=id,\,\,s\leq t \leq T & 
\end{eqnarray*}

Since $Q_0\geq 0$, $Q(t)\geq 0$, and $R^{-1}(t) >0$, it follows from the above equation that

\begin{eqnarray*}
\Pi _n (t) \geq 0,\mbox{  for all }n=1,2,...,\mbox{ and }t\in [0,T].
\end{eqnarray*}

Moreover, for all $t\in [0,T]$ and $n=2,3,...$

\begin{eqnarray}
0 \leq \Pi _{n+1} (t) \leq \Pi _n (t)\,.
\end{eqnarray}

To prove this, we notice that

\begin{eqnarray*}
&d\Pi _{n+1} (t)=(dt( F+\sigma _{11} \rho _2 (F_2 w) \rho _1 \rho _2 (F_1 w)+\sigma _{12}\rho _2 (F_2 w)F_2 w+&\\
&\sigma _{22}\rho _1 (F_1 w)\rho _2 \rho _1 (F_2 w)
+\sigma _{21 }\rho _1 (F_1 w) F_1 w- GR^{-1}G^* \Pi _n )(t)&\\
&+dM_1(t)( F_1 w)(t) + dM_2(t)( F_2 w )(t))^*\Pi _{n+1} (t)+&\\
&(dt ( \Pi _{n+1} F +\sigma_{11} \rho _1 \rho _2 (\Pi _{n+1} )\rho _1 (F_2 w)F_1 w&\\
&+\sigma _{12} \Pi _{n+1} \rho _2 (F_2 w)F_2 w
 + \sigma _{22} \rho _2 \rho _1 (\Pi _{n+1} )\rho _2(F_1 w)F_2 w &\\
&+\sigma _{21} \Pi _{n+1}  \rho _1 (F_1 w)F_1 w-\Pi _{n+1} GR^{-1}G^* \Pi _n )(t)&\\
&+dM_1(t)(\rho _1(\Pi _{n+1} ) F_1 w)(t) + dM_2(t)(\rho _2(\Pi _{n+1} ) F_2 w )(t))&\\
& +dt(Q+\Pi _n  GR^{-1}G^* \Pi _n)(t)
+dt(\sigma _{11}\rho _1 \rho _2( w^* F_1^* )\rho _1(\Pi  _{n+1})F_1 w &\\
&+\sigma _{12} w^* F_1^* \rho _2 (\Pi _{n+1} )  F_2 w 
 + \sigma _{22}\rho _2 \rho _1 (w^* F_2^* )\rho _2(\Pi _{n+1} )F_2 w&\\
&+\sigma_{21} w^* F_2^* \rho_1  ( \Pi _{n+1}  ) F_1 w )(t) &\\
&&\\
&\Pi _{n+1} (0)=Q_0\,\,, 0\leq t \leq T &
\end{eqnarray*}

Letting $P_n(t)=\Pi _{n+1} (t) - \Pi _{n} (t)$ the above yields

\begin{eqnarray*}
&dP _{n} (t)=(dt( F+\sigma _{11} \rho _2 (F_2 w) \rho _1 \rho _2 (F_1 w)+\sigma _{12}\rho _2 (F_2 w)F_2 w+&\\
&\sigma _{22}\rho _1 (F_1 w)\rho _2 \rho _1 (F_2 w)
+\sigma _{21 }\rho _1 (F_1 w) F_1 w- GR^{-1}G^* \Pi _n )(t)+&\\
&dM_1(t)( F_1 w)(t) + dM_2(t)( F_2 w )(t))^* P _{n} (t)&\\
&+(dt ( P _{n} F +\sigma_{11} \rho _1 \rho _2 (P _{n} )\rho _1 (F_2 w)F_1 w+&\\
&\sigma _{12} P _{n} \rho _2 (F_2 w)F_2 w
 + \sigma _{22} \rho _2 \rho _1 (P _{n} )\rho _2(F_1 w)F_2 w+&\\
&\sigma _{21} P _{n}  \rho _1 (F_1 w)F_1 w-P _{n} GR^{-1}G^* \Pi _n )(t)&\\
&+dM_1(t)(\rho _1(P_{n} ) F_1 w)(t) + dM_2(t)(\rho _2(P _{n} ) F_2 w )(t))&\\
&+dt(\sigma _{11}\rho _1 \rho _2( w^* F_1^* )\rho _1(P  _{n})F_1 w&\\ 
&+\sigma _{12} w^* F_1^* \rho _2 (P _{n} )  F_2 w  + \sigma _{22}\rho _2 \rho _1 (w^* F_2^* )\rho _2(P _{n} )F_2 w&\\
&+\sigma_{21} w^* F_2^* \rho_1  ( P _{n}  ) F_1 w )(t) -dt(P _{n-1}  GR^{-1}G^* P _{n-1})(t) &\\
&&\\
&P _{n} (0)=0\,\,, 0 \leq t \leq T &
\end{eqnarray*}

Thus as we did for $\Pi (t)$, for all $n=2,3,...$ and  $t \in [0,T]$ 

\begin{eqnarray*}
P _{n} (t)=-\int_0^t\,ds\,K_n(t,s)( P _{n-1} GR^{-1}G^* P _{n-1} )(s)K_n(t,s)^*
\end{eqnarray*}

weakly on $D$. Since $R^{-1}(s) >0$ for all $s\in [0,T]$, this implies that 

\[
P _{n} (t) \leq 0
\]

 thus proving (4.3). By (4.3) $\{ \Pi _n (t) \} _{n=1}^{+\infty}$ converges strongly on the invariant domain $D$ and the convergence is uniform on compact $t$-intervals. Let $\Pi=\{ \Pi  (t) \} _{0 \leq t \leq T}$ denote the limit process. Being a strong limit of a decreasing sequence of adapted, strongly continuous, positive processes, $\Pi$ has the same properties. By the uniformity of the convergence of the defining sequence  $\{ \Pi _n (t) \} _{n=1}^{+\infty}$ and the arbitrariness of $T$, $\Pi$ is locally bounded. As above we can show that for $n=1,2,...$ 

\begin{eqnarray*}
&\Pi _{n+1} (t)=\Phi (t,0)\,Q_0\,\Phi (t,0)^*+\int_0^t ds\,\Phi(t,s)(Q -\Pi _n GR^{-1}G^* \Pi _n &\\
&&\\
&+ P_n GR^{-1}G^* P_n)(s)\Phi (t,s)^*&
\end{eqnarray*}

weakly on the invariant domain $D$, where $\Phi(t,s)$ is the locally bounded solution of

\begin{eqnarray*}
&d\Phi (t,s)=[dt( F+\sigma _{11} \rho _2 (F_2 w) \rho _1 \rho _2 (F_1 w)+\sigma _{12} 
\rho _2 (F_2 w)F_2 w+&\\
&\sigma _{22} 
\rho _1 (F_1 w)\rho _2 \rho _1 (F_2 w)
+\sigma _{21 }\rho _1 (F_1 w) F_1 w )(t)&\\
&+dM_1(t)( F_1 w)(t)\ + dM_2(t)( F_2 w )(t)]^* \Phi (t,s)&\\
&&\\
&\Phi (s,s)=id,\,\,s\leq t \leq T. &
\end{eqnarray*}

Let $h,\xi \in D$. By the uniformity of the convergence of $\Pi _n (t) \rightarrow \Pi (t)$ and \break $P _n (t) \rightarrow 0$, and by the local boundedness of $\Pi$, upon letting $n\rightarrow +\infty$ we obtain by the bounded convergence theorem and

\begin{eqnarray}
&\langle \Pi _{n+1} (t) h, \xi \rangle=\langle [\Phi (t,0)\,Q_0\,\Phi (t,0)^*+\int_0^t ds\,\Phi(t,s)(Q -\Pi _n GR^{-1}G^* \Pi _n& \nonumber\\
&+ P_n GR^{-1}G^* P_n)(s)\Phi (t,s)^* ],h, \xi \rangle&\nonumber
\end{eqnarray}

that

\begin{eqnarray*}
\langle \Pi  (t) h, \xi \rangle=\langle [\Phi (t,0)\,Q_0\,\Phi (t,0)^*+\int_0^t ds\,\Phi(t,s)(Q -\Pi  GR^{-1}G^* \Pi)(s)\Phi (t,s)^* ],h, \xi \rangle
\end{eqnarray*}

from which, by the arbitrariness of $h, \xi$ we have that

\begin{eqnarray*}
 \Pi  (t) =\Phi (t,0)\,Q_0\,\Phi (t,0)^*+\int_0^t ds\,\Phi(t,s)(Q -\Pi GR^{-1}G^* \Pi)(s)\Phi (t,s)^* 
\end{eqnarray*}

for all $t \in [0,T] $, weakly on $D$. By taking the differential of both sides of the above we can show that $\Pi$ solves (4.1). To see that such $\Pi$ is unique, let $\hat \Pi$ be another solution of (4.1). Letting $P(t)=\Pi(t)-\hat \Pi (t)$ we obtain

\begin{eqnarray*}
&dP (t)=[dt( F+\sigma _{11} \rho _2 (F_2 w) \rho _1 \rho _2 (F_1 w)+\sigma _{12} 
\rho _2 (F_2 w)F_2 w+&\\
&\sigma _{22}\rho _1 (F_1 w)\rho _2 \rho _1 (F_2 w)+\sigma _{21 }\rho _1 (F_1 w) 
F_1 w- GR^{-1}G^* \hat  \Pi )(t)&\\
&+dM_1(t)( F_1 w)(t) + dM_2(t)( F_2 w )(t)]^* P (t)+&\\
&[dt ( P F +\sigma_{11} \rho _1 
\rho _2 (P )\rho _1 (F_2 w)F_1 w+&\\
&\sigma _{12} P \rho _2 (F_2 w)F_2 w + \sigma _{22} \rho _2 \rho _1 (P )\rho _2(F_1 w)
F_2 w +&\\
&\sigma _{21} P  \rho _1 (F_1 w)F_1 w-P GR^{-1}G^* \hat \Pi  )(t)&\\
&+dM_1(t)(\rho _1(P ) F_1 w)(t) + dM_2(t)(\rho _2 (P) F_2 w )(t)]+&\\
&dt[\sigma _{11}\rho _1 \rho _2( w^* F_1^* )\rho _1(P )F_1 w+\sigma _{12} w^* F_1^* \rho _2 (P )  F_2 w  +&\\
& \sigma _{22}\rho _2 \rho _1 (w^* F_2^* )\rho _2(P )F_2 w+\sigma_{21} w^* F_2^* \rho_1  ( P) F_1 w ](t)&\\
& -dt(P GR^{-1}G^* P )(t)&\\
&&\\
&P (0)=0\,\,, 0\leq t \leq T .&
\end{eqnarray*}

Thus, as before, $P (t) \leq 0$ i.e $\Pi (t) \leq \hat \Pi (t)$. By interchanging $\Pi$ and $\hat \Pi$ in \break $P(t)=\Pi(t)-\hat \Pi (t)$ and replacing $\hat \Pi$ by $\Pi$ in the above equation we obtain that \break $\hat \Pi (t) \leq \Pi (t)$. Thus $\Pi(t)=\hat \Pi (t)$ which proves uniqueness.

 We now turn to the case of the Riccati equation corresponding to (3.2) and (3.6) which will be treated by using the, just proved, case corresponding to (3.1) and (3.5) and reversing the time flow. So let $s=T-t$ in the Boson version of (4.1) and let, for an operator process $K$, $\hat K (s)=K(T-s)$ to obtain

\begin{eqnarray*}
&d\hat \Pi (s)=ds((\hat F+\tilde \sigma _{11} \rho _2 (\hat F_2 \hat w) \rho _1 \rho _2 (\hat F_1 \hat w)+\tilde \sigma _{12} 
\rho _2 (\hat F_2 \hat w)\hat F_2 \hat w+&\\
&\tilde \sigma _{22} 
\rho _1 (\hat F_1 \hat w) \rho _2 \rho _1 (\hat F_2 \hat w)
+\tilde \sigma _{21 }\rho _1 (\hat F_1 \hat w) \hat F_1 \hat w)^* \hat \Pi + &\\
&( \hat \Pi \hat F +\tilde \sigma_{11} 
 \rho _1 \rho _2 (\hat \Pi )\rho _1 (\hat F_2 \hat w)\hat F_1 \hat w+&\\
&\tilde \sigma _{12} \hat \Pi \rho _2 (\hat F_2 \hat w)\hat F_2 \hat w
+\tilde \sigma _{22} \rho _2 \rho _1 (\hat \Pi )\rho _2(\hat F_1 \hat w)\hat F_2 \hat w+&\\
&\tilde \sigma _{21} \hat \Pi  \rho _1 (\hat F_1 \hat w) \hat F_1 w))+\tilde \sigma _{11} \rho _1 \rho _2( {\hat w}^*  {\hat F}_1^* )\rho _1(\hat \Pi ) \hat F_1 \hat w &\\
&+\tilde  \sigma _{12} {\hat w}^* {\hat F}_1^* \rho _2 (\hat \Pi )\hat  F_2
 \hat  w  +\tilde \sigma _{22} 
\rho _2 \rho _1 ({\hat w}^* {\hat F}_2^* )\rho _2(\hat \Pi )\hat F_2 \hat w&\\
&+\tilde \sigma_{21}{\hat w}^* {\hat F}_2^* \rho_1  (\hat \Pi)\hat F_1 \hat w +\hat Q -\hat \Pi \hat G {\hat R}^{-1} {\hat G}^* \hat \Pi )(s) &\\
&+  dN_1(s)(\rho _1 ({\hat w}^* {\hat F}_2^* )\hat \Pi +  \rho _1 (\hat \Pi) \hat F_1 \hat w)(s)& \\
&+  dN_2(s)(\rho _2({\hat w}^* {\hat F}_1^* )\hat \Pi+  \rho _2 (\hat \Pi )\hat F_2 \hat w)(s)& \\
&&\\
&\hat \Pi (0)=Q_T\,\,, 0\leq s \leq T& 
\end{eqnarray*}

where the Levy-pair $(N_1,N_2)$ is defined by 

\begin{eqnarray*}
N_1(s)=-M_1(T-s),\,\,\,N_2(s)=-M_2(T-s)
\end{eqnarray*}

with corresponding It\^{o} table

 \begin{eqnarray*}
dN_b^* (s)\,dN_a(s)=\tilde {\sigma}_{ba} (s)\,ds
\end{eqnarray*}

where $a,b \in \{1,2\}$ and
 
\begin{eqnarray*}
\tilde {\sigma}_{ba} (s)=-\sigma_{ba} (T-s).
\end{eqnarray*}

Since the above differential equation is of the same form as the equation studied in the first part of this proof, the proof is complete.
\end{proof}

\medskip
 \section{\textbf{  Control of quantum flows}}
\medskip

To illustrate the use of the results of the previous section in the control of quantum flows, we consider a quantum flow $\{j_t(X)/\,t\in [0,T]\}$ of bounded linear operators on  $H_0\otimes \Gamma$ defined by $j_t(X)=U_t^*\,X \,U_t$ where, following \cite{o}, $H_0$ is a  separable Hilbert space, $\Gamma$ is the Boson Fock space over $ L^2 ([0,+\infty ))$, $X$ is a self-adjoint operator on $H_0$ identified with its ampliation $X \otimes I$ to $H_0\otimes \Gamma$, and $U=\{U_t \, / \, t \geq 0 \}$ is a unitary process satisfying on  $H_0 \otimes \Gamma$ a quantum stochastic differential equation of the form

\begin{equation}
dU_t=-((iH+\frac{1}{2}\,L^*L)\,dt+ L^* dA_t -L\, dA_t^{\dagger})\,U_t,\,t\in [0,T]
\end{equation}

 with adjoint

\begin{equation}
dU_t^*=-U_t^*\,((-iH+\frac{1}{2}\,L^*L)\,dt- L^* dA_t +L\, dA_t^{\dagger}),\,t \in [0,T]
\end{equation}

and

\[
U_0=U_0^*=I
\]

where $H,\,L$ are bounded operators on $H_0$ with $H$ self-adjoint. The  "annihilation" and "creation" processes $A=\{A_t \,/\, t \geq 0\}$ and  $A^{\dagger}=\{A_t^{\dagger}\,/\, t \geq 0\}$ that drive the above equations are an example of a Boson Levy pair of the type described in Section 4.

Using quantum It\^{o}'s formula for first order white noise, namely $dA_t\,dA_t^{\dagger}=dt$ and all other products of differentials are equal to zero, we can show  that the flow $\{j_t(X)/\,t\in [0,T]\}$ satisfies the  quantum stochastic differential equation 

\begin{eqnarray*}
&dj_t(X)=j_t(i[H,X]-\frac{1}{2}(L^*LX+XL^*L-2L^*XL))\,dt &\\
&+j_t([L^*,X])\,dA_t +j_t([X,L])\,dA_t^{\dagger}&
\end{eqnarray*}

with

\[
j_0(X)=X,\,t\in [0,T]
\]

\begin{definition}
For  any  finite time interval $[0,T]$ and   any vector $\xi$ in the exponential domain of $H_0\otimes \Gamma$ 

 \begin{equation}
\hat J_{\xi,T}(L)=\int_0^T\,[\,\|j_t(X)\xi\|^2+\frac{1}{4}\|j_t(L^*L)\xi\|^2\, ]\,dt+\frac{1}{2}\|j_T(L)\xi\|^2
\end{equation}

where $L$ is as in (5.1).

\end{definition}

 Thinking of $L$ as a control, as pointed out in the introduction we interpret the first term of the right hand side of (5.3) as a measure of the size of the flow over $[0,T]$, the second as a measure of the control effort over $[0,T]$ and the third as a "penalty" for allowing the evolution to go on for a long time. We consider the problem of controlling the size of such a flow by minimizing the performance functional $\hat J_{\xi,T}(L)$ of (5.3). 

\begin{theorem}
  Let $U=\{U_t\,/\,t\geq 0\}$ be an adapted process satisfying the  quantum stochastic differential equation 

\begin{equation}
dU_t=(F_tU_t+u_t)\,dt+ \Psi_t \,U_t\, dA_t+ \Phi_t \,U_t\, dA_t^{\dagger},\,U_0=I,\,t\in [0,T]
\end{equation}

where $T > 0$ is a fixed finite horizon, $dA_t^{\dagger}$ and $dA_t$ are the differentials of the creation and annihilation processes  of \cite{o},  and the coefficient processes are adapted, bounded, strongly continuous and square integrable processes living on the exponential  domain $\mathcal{E}=$ span $\{\xi=\xi_0 \otimes \psi (f)\}$ of the tensor product $H_0\otimes \Gamma$ of a system (separable Hilbert) space $H_0$ and the noise (Boson  Fock) space $\Gamma$ on $L^2([0,T],\mathbb{C})$. 

The quadratic performance functional

\begin{equation}
\tilde J_{\xi,T}(u)=\int_0^T\,[<U_t \xi,X^*X\,U_t \xi>+<u_t \xi,u_t \xi>]\,dt+<U_T \xi,M\,U_T \xi>
\end{equation}

where $X,\,M$ are  bounded operators on $H_0$, identified with their ampliations \break 
$X \otimes I,\,M \otimes I$ to $H_0 \otimes \Gamma$, with $M \geq 0$,  is minimized by the feedback control process $u_t=-\Pi_tU_t$, where the bounded, positive, self-adjoint process $\{\Pi_t\,/\,t\in [0,T]\}$, with $\Pi_T=M$, is the solution of the quantum stochastic Riccati equation

\begin{eqnarray}
&d\Pi_t+(\Pi_tF_t+F^*_t\Pi_t+\Phi_t^*\Pi_t\Phi_t-\Pi_t^2+X^*X )\,dt+&\\
& (\Pi_t\Psi_t+\Phi_t^*\Pi_t)\, dA_t+(\Pi_t\Phi_t+\Psi_t^*\Pi_t)\, dA_t^{\dagger}=0 &\nonumber
\end{eqnarray}

 and the minimum value is $<\xi,\Pi_0 \xi>$.

\end{theorem}

\begin{proof} The proof follows by a direct translation of the results of Theorem 1 related to (3.1) and (3.5) in the framework of (5.1) and (5.2) 

\end{proof}

\begin{lemma}In the notation of Definitions 5, 1, and Theorem 3, 
 if  $X$  is self-adjoint, $M=\frac{1}{2}L^*L$,  and $u_t=-\frac{1}{2}L^*L\,U(t)$  for all $t \geq 0$, where $U=\{U(t)\,/\,t \geq 0 \}$ is the unitary solution of (5.1),  then

\[
\hat J_{\xi,T}(L)=J_{\xi,T}(u)=\tilde J_{\xi,T}(u)
\]

\end{lemma}

\begin{proof} By (5.5), for self-adjoint $X$

\begin{eqnarray*}
&\tilde J_{\xi,T}(u)=&\\
&\int_0^T \,[< \xi,U_t^*X X U_t \xi>+\|u_t \xi\|^2]\,dt+< \xi,U_T^*M U_T \xi>&\\
&=\int_0^T\,[< \xi,U_t^*X U_t U_t^*  X U_t \xi>+\|u_t \xi\|^2]\,dt+<\xi,j_T(M)\xi>&\\
& =\int_0^T\,[\|j_t(X)\xi\|^2+\|u_t \xi\|^2]\,dt+<\xi,j_T(M)\xi> &\\
&=J_{\xi,T}(u)&\\
&=\int_0^T\,[\|j_t(X)\xi\|^2+\frac{1}{4}\|L^*LU_t \xi\|^2]\,dt+\frac{1}{2}< \xi,U_T^*L^*LU_T \xi>   &\\
&=\int_0^T\,[\|j_t(X)\xi\|^2+\frac{1}{4}< \xi,U_t^*(L^*L)^2U_t \xi>]\,dt+\frac{1}{2}< \xi,U_T^*L^*LU_T \xi>&\\
&=\int_0^T\,[\|j_t(X)\xi\|^2+\frac{1}{4}< \xi,U_t^*(L^*L)U_tU_t^* (L^*L)U_t \xi>]\,dt+\frac{1}{2}< \xi,U_T^*L^*LU_T \xi>&\\
&=\int_0^T\,[\|j_t(X)\xi\|^2+\frac{1}{4}\|j_t(L^*L)\xi\|^2]\,dt+\frac{1}{2}<\xi,j_T(L^*L)\xi>&\\
&=\hat J_{\xi,T}(L)&
\end{eqnarray*}

\end{proof}

\begin{theorem}

 Let $\xi$ be a vector in the exponential vectors domain $\mathcal{E}$ of $H_0 \otimes \Gamma$, let $0<T<+\infty$, and let $H$, $L$, $X$ be bounded linear operators on $H_0$ such that $H$, $X$  are self-adjoint and the pair ($\frac{i}{2}H$, $X $) is stabilizable. The quadratic performance criterion $\hat J_{\xi,T}(L)$ of Definition 5, associated with the quantum stochastic flow $\{j_t(X)=U_t^*\,X \,U_t\,/\,t \geq 0\}$  where  $U=\{U_t\,/\,t\geq 0\}$ is the solution of (5.1), is minimized by

\begin{equation}
L=\sqrt{2}\,\Pi_{\infty}^{1/2}\,W\,\,\,(\mbox{polar decomposition of}\,\,L)  
\end{equation}

where  $\Pi_{\infty}$ is a positive self-adjoint solution of the "algebraic Riccati equation"

\begin{eqnarray}
\frac{i}{2}[H,\Pi_{\infty}]+\frac{1}{4}\Pi_{\infty}^2+X^2=0  
\end{eqnarray}

 and $W$ is any bounded unitary linear operator on $H_0$ commuting with  $\Pi_{\infty}$.

Moreover

\begin{eqnarray}
\min_{L}\,J_{\xi,T}(L)= <\xi,\Pi_{\infty}\xi>
\end{eqnarray}

 independent of $T$.

\end{theorem}

 \begin{proof}Looking at (5.1) as (5.4) with $u_t=-\frac{1}{2}L^*LU_t$ and taking $M=\frac{1}{2}L^*L$, in view of Lemma 1 (5.5) yields  

\begin{eqnarray*}
\tilde J_{\xi,T}(u)=\hat J_{h,T}(L)=\int_0^T\,[\,\|j_t(X)h\|^2+\frac{1}{4}\|j_t(L^*L)h\|^2\, ]\,dt+\frac{1}{2}\|j_T(L)h\|^2
\end{eqnarray*}

 By Theorem 3, in order for $L$ to be optimal it must satisfy

\begin{eqnarray*}
\frac{1}{2}L^*L=\Pi_t 
\end{eqnarray*}

where $\Pi_t$ is the solution of (5.6) for $F_t=-iH$, $\Phi_t=L$ and $\Psi_t=-L^*$. For these choices (5.6) reduces, by the time independence of $\Pi_t$ and the linear independence of $dt$, $dA_t$ and $dA_t^{\dagger}$ , to the equations

\begin{equation}
[L,L^*]=0\,\,(\mbox{i.e $L$ is normal})  
\end{equation}

and

\begin{eqnarray}
\frac{i}{2}[H,\Pi_{\infty}]+\frac{1}{4}\Pi_{\infty}^2+X^2=0  
\end{eqnarray}

where $[ \cdot, \cdot ]$ denotes the usual operator commutator and

\begin{equation}
\Pi_{\infty}=\frac{1}{2}L^*L
\end{equation}

 We recognize (5.11) as a special case of the algebraic Riccati equation (ARE) (see \cite{h}). It is known that if there exists a bounded linear operator $K$ on $H_0$ such that $\frac{i}{2}H+KX$ is the generator of an asymptotically stable semigroup (i.e if the pair ($\frac{i}{2}H$, $X
$) is stabilizable) then (5.11) has a positive self-adjoint solution $\Pi_{\infty}$.  Now (5.7) follows by (5.10) and (5.12) in conjunction with Lemma 1. 

\end{proof}

\end{document}